\documentclass[twocolumn,preprintnumbers,amsmath,amssymb]{revtex4-2}
\usepackage[english]{babel}
\usepackage{enumitem}
\usepackage{epsfig}
\usepackage{graphicx}
\usepackage{dcolumn}
\usepackage{bm}
\usepackage{amsmath,amsfonts}
\usepackage{amssymb}
\usepackage{amsthm}
\usepackage{physics}
\usepackage{comment}
\usepackage{color} 
\usepackage{xcolor}
\usepackage[final,colorlinks,linkcolor=blue,anchorcolor=blue,urlcolor=blue,citecolor=blue]{hyperref}
\usepackage{caption}
\usepackage{subcaption}
\usepackage{booktabs}
\usepackage{siunitx}
\usepackage{epstopdf}
\usepackage{diagbox}
\usepackage{framed}
\usepackage{dsfont}
\usepackage{float}
\newtheorem{proposition}{Proposition}

\begin{document}
	\title{Device-independent quantum cryptography with input leakage}
    \author{Víctor Zapatero$^{1,2,3}$}
	\author{Marcos Curty$^{1,2,3}$}
 
\affiliation{$^{1}$Vigo Quantum Communication Center, University of Vigo, Vigo E-36310, Spain}

\affiliation{$^2$Escuela de Ingeniería de Telecomunicación, Department of Signal Theory and Communications, University of Vigo, Vigo E-36310, Spain}

\affiliation{$^3$AtlanTTic Research Center, University of Vigo, Vigo E-36310, Spain}

\begin{abstract}
Device-independence is the gold standard of quantum cryptography. To meet this standard, a central assumption is that no information leakage occurs during protocol execution. We relax this assumption by analyzing CHSH-based randomness certification and key distribution with partial leakage of the inputs, modeled in terms of a noisy channel. Our results quantify the certifiable local randomness and the secret key rate as a function of the magnitude of the input leakage.
\end{abstract}

\maketitle
\section{Introduction}
The possibility of quantifying the randomness of a physical process with respect to a knowledgeable eavesdropper is an appealing cryptographic feature of quantum theory~\cite{Law}. Bell nonlocality~\cite{Bell,Brunner_review,Scarani_book}, in particular, has been proven to push the necessary assumptions of randomness certification and amplification to their ultimate logical minimum~\cite{Pironio2010,ColbeckRenner2012,EkertRenner2014,Ramanathan2016,Brandao2016,KesslerFriedman}. Moreover, this phenomenon underlies the celebrated \emph{device-independent} (DI) approach to quantum key distribution (QKD)~\cite{E91,Mayers,Barrett2005,Acín2007,PironioDIQKD,reviewDI}, successfully demonstrated in a series of pioneering experiments~\cite{Lu,Nadlinger,Zhang,Liu}.

Fundamentally for DI cryptography, the devices must operate at secure locations that do not leak sensitive information, and the possibility of relaxing this assumption comes at the price of characterizing the leakage to a certain extent. In this work, we tackle the problem of DI cryptography with partial leakage of the protocol inputs. In practice, this leakage could be mediated by a hardware side-channel, such as electromagnetic radiation or power consumption if the inputs are produced with physical random number generators~\cite{Agrawal,Lavaud,Kocher}.

To quantify the potential effect of input leakage in DI cryptography, we adopt a standard adversarial viewpoint. First, we study the task of certifying randomness in a CHSH test~\cite{CHSH} with input leakage, generally described within the framework of~\emph{measurement-dependence}~\cite{BarrettGisin2011,Hall2011,Koh2012,Putz2014,Putz2016,Friedman2019,Hall2020,Supic2020}: in every round of the test, independent leakage variables correlated to the parties' inputs leak to the environment, and in particular to the adversary. Technically, this provides an example of \emph{causal}~\cite{Hall2020} or \emph{independent-sources}'~\cite{Putz2016} measurement-dependence~\cite{lightcones}. Interestingly, a different perspective on this problem has been adopted in a recent work~\cite{Ramanathan}, where the influence of each party's input on the other party's output is not constrained to be mediated by a measurement-dependent hidden-variable. Instead, the authors further relax the assumption of \emph{parameter-independence}~\cite{Scarani_book,Hall2020}, following a proposal of~\cite{Vieira2025} and resorting to a technical tool first presented in~\cite{crosstalk}. For a brief summary of Bell's assumptions, the reader is referred to Appendix~\ref{HVM}.

As a second contribution, we build on our results on randomness certification to analyze the security of DIQKD with input leakage, a possibility suggested in a number of works (see e.g.~\cite{Kofler,PironioDIQKD,Koh2012}) but never actually addressed in the literature. It is worth mentioning, however, that a security proof incorporating constrained leakage of the outputs is put forward in~\cite{Tan}.

The paper is organized as follows. In Sec.~\ref{preliminaries}, we formalize the canonical problem of DI randomness certification. In Sec.~\ref{DIRG}, we generalize this problem to incorporate input leakage, and in Sec.~\ref{DIQKD}, we extend the analysis to DIQKD. Finally, we present conclusions and directions for further work in Sec.~\ref{Outlook}, and the appendices contain additional results and supplemental material.

\section{Preliminaries}\label{preliminaries}
We first recall the problem of certifying randomness in a CHSH test~\cite{CHSH}, introducing the necessary background on Bell nonlocality along the way. Following~\cite{Scarani_book}, we refer to this task as DI randomness generation (DIRG). 

Throughout this work, upper-case letters denote random variables and lower-case letters denote their realizations. A CHSH test is an experiment with two non-communicating parties, Alice and Bob, each holding a box with two inputs and two outputs. In every round of the experiment, the parties locally select respective inputs $X$ and $Y$, and the boxes produce local outputs $A$ and $B$ in response. For simplicity, we assume that the boxes behave identically in every round, thus being described by a fixed list of probabilities $\{\mathcal{P}(a,b|x,y)\}_{a,b,x,y}$ experimentally accessible in the long run. We refer to this list as the \emph{observed behavior}, and we denote it by $\mathcal{P}$.

Generally speaking, the purpose of DIRG is to quantify, on the sole basis of $\mathcal{P}$, the unpredictability of the outputs for an adversary (Eve) holding a finer description of $\mathcal{P}$. In full generality for the considered i.i.d. scenario, this finer description is a \emph{hidden variable} $\Lambda$ shared between the boxes and Eve. Crucially, Eve does not need to interact with the boxes to agree on $\Lambda$: in every round of the test, all three can consistently draw $\Lambda$ from any distribution $q_{\lambda}$ using \emph{pre-shared randomness}. As long as $\Lambda$ is independent of $X$ and $Y$ ---in the jargon, \emph{measurement-independent}---, the observed behavior fulfills~\cite{clarification}
\begin{equation}\label{convex}
\mathcal{P}=\sum_{\lambda}q_{\lambda}\mathcal{P}_{\lambda},
\end{equation}
for a series of underlying behaviors $\mathcal{P}_{\lambda}$ associated to the different $\lambda$'s. To derive non-trivial bounds on Eve's predictive power, additional assumptions are required on top of measurement-independence. For instance, it is occasionally assumed that the $\mathcal{P}_{\lambda}$ are \emph{quantum}, meaning that
\begin{equation}\label{realization}
\mathcal{P}_{\lambda}(a,b|x,y)=\Tr\left[\left(A_{a|x}^{\lambda}\otimes{}B_{b|y}^{\lambda}\right)\rho_{\rm AB}^{\lambda}\right]
\end{equation}
for local measurements $\{A_{a|x}^{\lambda}\}$ and $\{B_{b|y}^{\lambda}\}$, and quantum states $\rho_{\rm AB}^{\lambda}$ of Eve's choice. We refer to the set of quantum behaviors ---i.e. those compliant with Eq.~(\ref{realization})--- as the \emph{quantum set} $\mathcal{Q}$. One can easily show that $\mathcal{Q}$ is convex, meaning that $\mathcal{P}$ is quantum if all the $\mathcal{P}_{\lambda}$ are quantum. When making this extra assumption of \emph{quantumness}, an additional layer of generality consists of assuming that $\rho_{\rm AB}^{\lambda}=\Tr_{\rm E}\left[\rho_{\rm ABE}^{\lambda}\right]$, where the quantum system E is under Eve's control. This is the adversary model that we adopt in this work. For obvious reasons, we refer to $\Lambda$ as Eve's \emph{classical side-information} and to E as Eve's \emph{quantum side-information}.

Crucially, the quantumness of $\mathcal{P}$ does not imply the unpredictability of the outcomes by any means: in general, this assumption is compatible with the possibility that $\Lambda$ pre-allocates the boxes' outputs based on their local inputs through deterministic \emph{response functions} $\mathcal{P}_{\lambda}(a|x)=\delta_{a,a_{x}^{\lambda}}$ and $\mathcal{P}_{\lambda}(b|y)=\delta_{b,b_{y}^{\lambda}}$~\cite{Fine}, such that
\begin{equation}\label{local}
\mathcal{P}(a,b|x,y)=\sum_{\lambda}q_{\lambda}\delta_{a,a_{x}^{\lambda}}\delta_{b,b_{y}^{\lambda}}.
\end{equation}
Any $\mathcal{P}$ satisfying Eq.~(\ref{local}) is called \emph{local}, and the set of all local behaviors is known as the \emph{local set} $\mathcal{L}$~\cite{caveat}. Obviously, no randomness can be certified if $\mathcal{P}\in\mathcal{L}$. Fortunately however, $\mathcal{Q}$ is a strict superset of $\mathcal{L}$, and the only deterministic quantum behaviors are those in $\mathcal{L}$. 

The randomness in the outcomes of $\mathcal{P}$ can be quantified in different (generally inequivalent) ways, but the central and most studied question goes as follows~\cite{Pironio2010,Pironio2013,Pironio2025}. Given an observed behavior $\mathcal{P}$, what is Eve's \emph{guessing probability} of $A$ for a fixed local input, say $X=x^{*}$? Eve can finely tune $\Lambda$ and the different $\mathcal{P}_{\lambda}$ to maximize her predictive power on this particular outcome. For the same purpose, she can perform a carefully selected measurement $\{E_{c}^{\lambda}\}$ on her quantum system E, possibly dependent on the side-information $\Lambda$. Note, however, that whether this measurement takes place before or after Alice's and Bob's is irrelevant due to commutativity, meaning that one can view Eve's measurement upon obtaining the outcome $c$ as preparing the subnormalized state $\tilde{\rho}^{\hspace{.05cm}\lambda{}c}_{\rm AB}=\Tr_{\rm E}\bigl[\bigl(\mathds{1}_{\rm AB}\otimes{}E_{c}^{\lambda}\bigr)\rho_{\rm ABE}^{\lambda}\bigr]$. This proves that, for the standard figure of merit under consideration, quantum-side information does not increase Eve's capabilities, and one can restrict the analysis to the classical side-information $\Lambda$ without loss of generality. Under these circumstances, Eve simply guesses for the most likely outcome based on $\Lambda$. That is to say, denoting by $\mathcal{G}_{a}$ the set of $\mathcal{P}_{\lambda}$'s whose largest marginal for $X=x^{*}$ is $\mathcal{P}_{\lambda}(a|x^{*})=\Tr\bigl[ (A_{a|x^{*}}^{\lambda}\otimes{}\mathds{1}_{\rm B})\rho^{\lambda}_{\rm AB}\bigr]$, the requested upper bound on the guessing probability is formally given by
\begin{equation}\begin{split}\label{optimization_full}
&\max_{\displaystyle{\{q_{\lambda},\mathcal{P}_{\lambda}\}}}\quad \sum_{a}\sum_{\mathcal{P}_{\lambda}\in\mathcal{G}_{a}}q_{\lambda}\mathcal{P}_{\lambda}(a|x^{*})\hspace{.4cm}\\
&\hspace{.2cm}\textrm{s.t.}\hspace{.2cm}\sum_{\lambda}q_{\lambda}\mathcal{P}_{\lambda}=\mathcal{P},\hspace{.1cm}\mathcal{P}_{\lambda}\in{}\mathcal{Q}\hspace{.2cm}\textrm{for all }\lambda.
\end{split}\end{equation}
Eq.~(\ref{optimization_full}) uses all 16 entries of $\mathcal{P}$ as constraints of the optimization~\cite{Bancal,Nieto}, but a standard approach is to only test the value of the CHSH functional~\cite{CHSH}
\begin{equation}\label{def}
S[\mathcal{P}]=\sum_{a,b,x,y}(-1)^{a+b+xy}\mathcal{P}(a,b|x,y),
\end{equation}
where $\mathcal{P}(a,b|x,y)=\sum_{\lambda}q_{\lambda}\mathcal{P}_{\lambda}(a,b|x,y)$ and we assume that $a,b,x,y\in\{0,1\}$. The resulting problem reads exactly as Eq.~(\ref{optimization_full}), but replacing the full behavior request, $\sum_{\lambda}q_{\lambda}\mathcal{P}_{\lambda}=\mathcal{P}$, by the observed CHSH value, $S[\mathcal{P}]=S_{\rm obs}$. Indeed, this constraint suffices because $\mathcal{L}$ is a \emph{polytope} and $S[\mathcal{P}]=2$ is one of its non-trivial \emph{facets}, such that any $S_{\rm obs}>2$ rules out membership of $\mathcal{P}$ in $\mathcal{L}$.

An apparent difficulty of this optimization task is the unrestricted character of the side-information $\Lambda$. Nevertheless, the convexity of $\mathcal{Q}$ and a symmetry of $S$ drastically simplify the problem, which can ultimately be cast as an optimization over a single behavior of pure two-qubit states and projective measurements~\cite{Scarani_book}. Analytically solving this optimization yields a maximum guessing probability of~\cite{Pironio2010}
\begin{equation}\label{pguess}
G_{x^{*}}(S_{\rm obs})=\frac{1}{2}\left[1+\sqrt{2-\left(\frac{S_{\rm obs}}{2}\right)^2}\right].
\end{equation}
In particular, complete unpredictability is granted for the Tsirelson's bound $S_{\rm obs}=2\sqrt{2}$~\cite{Tsirelson}, which \emph{self-tests} the behavior of a maximally entangled two-qubit state under suitable Pauli measurements up to local isometries~\cite{Supic_review}.

\section{DIRG with input leakage}\label{DIRG}
Let us now assume that, during the Bell experiment, partial information about the inputs is leaked to Eve. To describe the leakage, we consider a pair of independent hidden variables $U$ and $V$ generated in every round, where $U$ stands for Alice's leak, partially correlated to $X$, and $V$ stands for Bob's leak, partially correlated to $Y$. This is an instance of \emph{causal}~\cite{Hall2020} or \emph{independent-sources}~\cite{Putz2016} measurement-dependence. The situation is illustrated in Fig.~\ref{fig:depiction}.
\begin{figure}[H]
    \centering
    \includegraphics[scale=0.2]{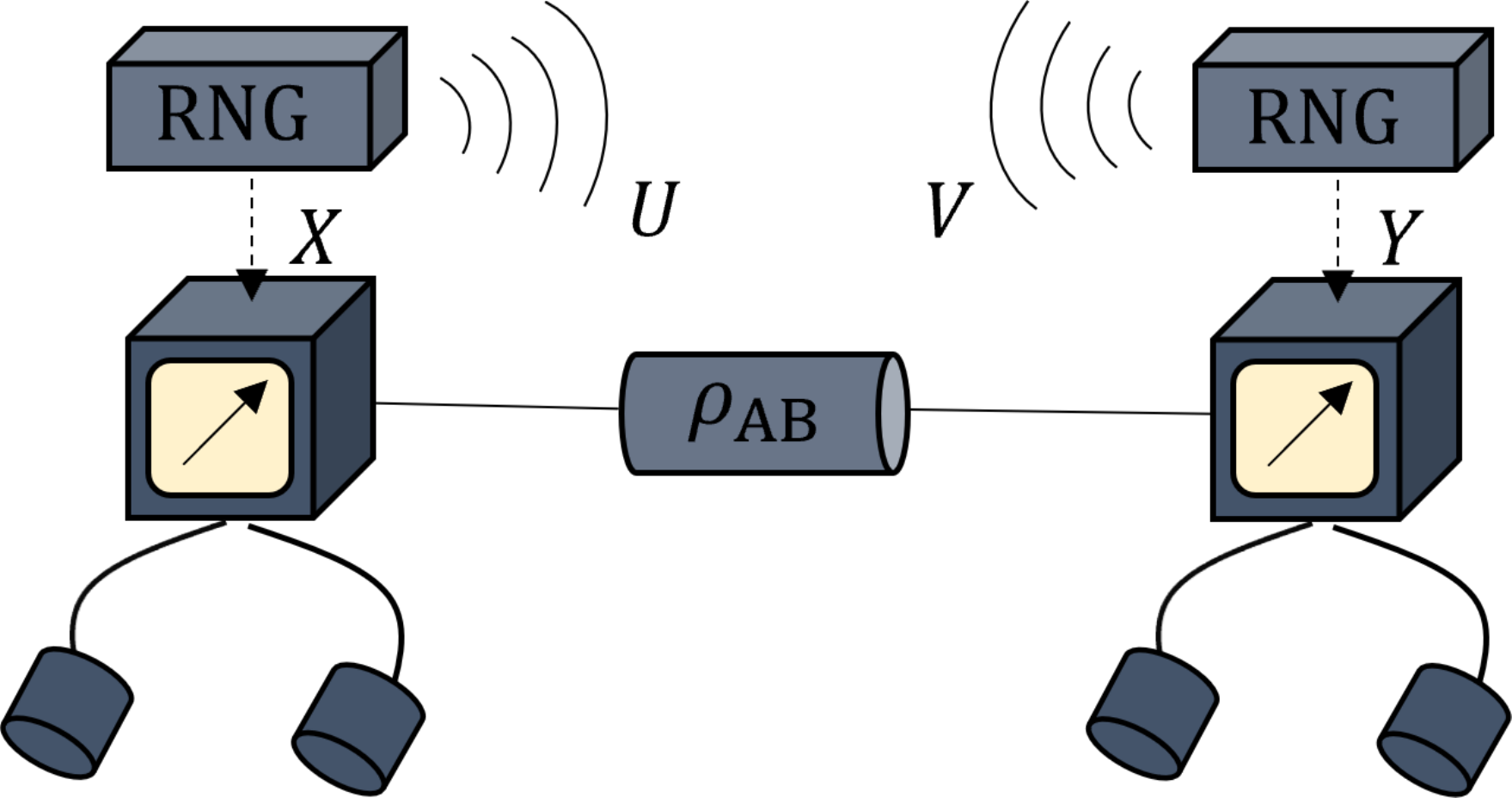}
    \caption{A CHSH test with input leakage. In the process of generating the inputs $X$ and $Y$, partial information about them is leaked, modeled with two independent leakage variables $U$ and $V$. Most prominently, if physical random number generators (RNGs) are used to select the inputs, $U$ and $V$ could correspond to hardware side channels, such as electromagnetic radiation or power consumption. If the test is to be used for cryptographic purposes, one can conservatively assume that the leakage is accessible to the adversary, the measurement devices, and the entanglement source ($\rho_{\rm AB}$).}
    \label{fig:depiction}
\end{figure}

To account for the impact of such leakage on DIRG, we shall assume that the pair $(U,V)$ is accessible to the boxes and the entanglement source, which collude to finely tune the quantum realization in every round. Paranoid as it may seem, this approach allows to quantify how much the leakage can affect DIRG \emph{in principle}, which provides a conservative estimate of how much it can affect it \emph{in practice}.

In full generality, upon receiving $(U,V)=(u,v)$, the source and the boxes locally draw suitable side-information $\lambda$ from a distribution $q_{\lambda|uv}$ using pre-shared randomness, and the tuple $(u,v,\lambda)$ determines the state to be delivered, $\rho_{\rm AB}^{uv\lambda}$, and the measurements contemplated on each side, $\{A_{a|x}^{uv\lambda}\}$ for Alice and $\{B_{b|y}^{uv\lambda}\}$ for Bob. In short, $(U,V)=(u,v)$ determines the behavior
\begin{equation}
\mathcal{P}_{uv}=\sum_{\lambda}q_{\lambda|uv}\mathcal{P}_{uv\lambda},
\end{equation}
where $\mathcal{P}_{uv\lambda}(a,b|x,y)=\Tr\bigl[\bigl(A_{a|x}^{uv\lambda}\otimes{}B_{b|y}^{uv\lambda}\bigr)\rho_{\rm AB}^{uv\lambda}\bigr]$~\cite{crucial}. Note that, in principle, the side-information $\Lambda$ could have an impact in the guessing probability, as we do not request each $\mathcal{P}_{uv}$ to be committed to a single guess, but we allow it to have contributions $\mathcal{P}_{uv\lambda}$ both in $\mathcal{G}_{0}$ and in $\mathcal{G}_{1}$. Putting it all together, for a given model of $(U,V)$, captured by the conditional statistics $q_{uv|xy}=q_{u|x}q_{v|y}$, Eve's guessing probability subject to $S[\mathcal{P}]=S_{\rm obs}$ is computed via
\begin{eqnarray}\label{program}
&&\max_{\displaystyle{\bigl\{q_{\lambda|uv},\mathcal{P}_{uv\lambda}\bigr\}}}\quad \sum_{u,v}q_{uv|x^{*}}\sum_{a}\sum_{\mathcal{P}_{uv\lambda}\in\mathcal{G}_{a}}q_{\lambda|uv}\mathcal{P}_{uv\lambda}(a|x^{*})\nonumber \\
&&\hspace{.2cm}\textrm{s.t.}\hspace{.2cm}\sum_{a,b,x,y}(-1)^{a+b+xy}\sum_{u,v}q_{uv|xy}\mathcal{P}_{uv}(a,b|x,y)=S_{\rm obs},\nonumber \\
&&\hspace{.9cm}\mathcal{P}_{uv\lambda}\in{}\mathcal{Q}\hspace{.1cm}\textrm{for all } u,v,\lambda,
\end{eqnarray}
where $\mathcal{P}_{uv}(a,b|x,y)=\sum_{\lambda}q_{\lambda|uv}\mathcal{P}_{uv\lambda}(a,b|x,y)$. To get rid of the coefficients $q_{\lambda|uv}$, we introduce the subnormalized behaviors
\begin{equation}\label{subnormalized}
\tilde{\mathcal{P}}_{uv}^{(a)}=\sum_{\mathcal{P}_{uv\lambda}\in\mathcal{G}_{a}}q_{\lambda|uv}\mathcal{P}_{uv\lambda}
\end{equation}
for all $u$, $v$ and $a$, where the subnormalized quantum set is defined as $\tilde{\mathcal{Q}}=\{\hspace{.01cm}w{}\mathcal{P}\hspace{.01cm}:\hspace{.01cm} 0\leq{}w\leq{}1,\mathcal{P}\in\mathcal{Q}\hspace{.02cm}\}$. Obviously, it follows from Eq.~(\ref{subnormalized}) that $\mathcal{P}_{uv}=\tilde{\mathcal{P}}_{uv}^{(0)}+\tilde{\mathcal{P}}_{uv}^{(1)}$ for all $u$ and $v$, and one can cast the problem as
\begin{equation}\begin{aligned}\label{reduction_1} &\max_{\displaystyle{\bigl\{\tilde{\mathcal{P}}_{uv}^{(a)}\bigr\}}}\quad \sum_{u,v}q_{uv|x^{*}}\left[\tilde{\mathcal{P}}_{uv}^{(0)}(0|x^{*})+\tilde{\mathcal{P}}_{uv}^{(1)}(1|x^{*})\right]\hspace{.4cm}\\ &\hspace{.2cm}\textrm{s.t.}\hspace{.2cm}\sum_{a,b,x,y}(-1)^{a+b+xy}\sum_{u,v}q_{uv|xy}\left[\tilde{\mathcal{P}}^{(0)}_{uv}(a,b|x,y)+\right.\\
&\hspace{.8cm}\left.\tilde{\mathcal{P}}^{(1)}_{uv}(a,b|x,y)\right]=S_{\rm obs},\\
&\hspace{.9cm}\tilde{\mathcal{P}}_{uv}^{(a)}\in{}\tilde{\mathcal{Q}}\hspace{.1cm}\textrm{for all } u,v,a.
\end{aligned}\end{equation}
This is an optimization task over 8 behaviors, but a standard trick allows to reduce this number to 4~\cite{Scarani_book}. Indeed, let us denote the bit flips of $a$ and $b$ by $a'$ and $b'$. Defining $\tilde{\mathcal{P}}_{uv}^{'(1)}(a,b|x,y)=\tilde{\mathcal{P}}_{uv}^{(1)}(a',b'|x,y)$, the behavior $\tilde{\mathcal{P}}_{uv}^{'(1)}$ fulfills: (i) $\tilde{\mathcal{P}}_{uv}^{(1)}\in{}\tilde{\mathcal{Q}}$ if and only if $\tilde{\mathcal{P}}_{uv}^{'(1)}\in{}\tilde{\mathcal{Q}}$, (ii) $\sum_{ab}(-1)^{a+b}\tilde{\mathcal{P}}_{uv}^{(1)}(a,b|x,y)=\sum_{ab}(-1)^{a+b}\tilde{\mathcal{P}}_{uv}^{'(1)}(a,b|x,y)$, and (iii) $\tilde{\mathcal{P}}_{uv}^{(1)}(1|x)=\tilde{\mathcal{P}}_{uv}^{'(1)}(0|x)$ for all $x$. As a consequence, Eq.~(\ref{reduction_1}) can be directly written in terms of the 4 normalized behaviors $\mathcal{P}_{uv}^{'}=\tilde{\mathcal{P}}_{uv}^{(0)}+\tilde{\mathcal{P}}_{uv}^{'(1)}$, such that Eve's guessing probability $G_{x^{*}}(S_{\rm obs})$ is computed via
\begin{eqnarray}\label{behavior_reduced}
&&\max_{\displaystyle{\bigl\{\mathcal{P}_{uv}^{'}\bigr\}}}\quad \sum_{u,v}q_{uv|x^{*}}\mathcal{P}_{uv}^{'}(0|x^{*})\hspace{.4cm}\nonumber \\
&&\hspace{.2cm}\textrm{s.t.}\hspace{.2cm}\sum_{a,b,x,y}(-1)^{a+b+xy}\sum_{u,v}q_{uv|xy}\mathcal{P}_{uv}^{'}(a,b|x,y)=S_{\rm obs},\nonumber \\
&&\hspace{.9cm}\mathcal{P}_{uv}^{'}\in{}\mathcal{Q}\hspace{.2cm}\textrm{for all } u,v.
\end{eqnarray}
One can tackle this problem by relaxing the membership conditions $\mathcal{P}_{uv}^{'}\in{}\mathcal{Q}$ into $\mathcal{P}_{uv}^{'}\in{}\mathcal{Q}^{k}$, where $\mathcal{Q}^{k}$ denotes the $k$-th level of the Navascués-Pironio-Acín (NPA) hierarchy~\cite{NPA1,NPA2}. In Appendix~\ref{intermediate}, we provide an explicit semidefinite program that computes $G_{0}(S_{\rm obs})$ using the intermediate level $\mathcal{Q}^{1+\mathrm{AB}}$. This level is tight in the absence of leakage, reproducing the analytical bound of Eq.~(\ref{pguess}) up to numerical precision. Notably, the program is plug-and-play for the leakage model, in the sense that it can be run for any model given its conditional statistics $q_{u|x}$ and $q_{v|y}$. As an example, in Fig.~\ref{fig:dirg} we plot $G_{0}(S_{\rm obs})$ for a binary symmetric leakage channel, setting a common crossover probability $P(U\neq{}X)=P(V\neq{}Y)=\varepsilon\leq{}1/2$ on each side.
\begin{figure}[H]
    \centering
    \includegraphics[width=\columnwidth]{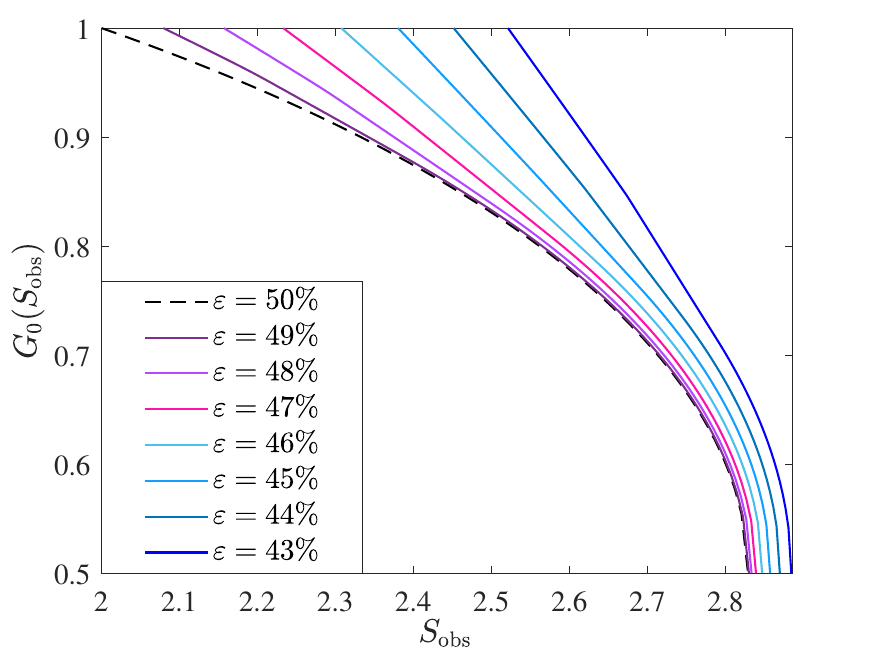}
    \caption{Local guessing probability of Eve in a CHSH test with input leakage, as a function of the observed CHSH value, $S_{\rm obs}$. For illustration purposes, the leakage is modeled as the output of a binary symmetric channel acting on the local input. Several crossover parameters $\varepsilon$ are considered, defined as the probability of a mismatch between the local input and its leaked value.}
    \label{fig:dirg}
\end{figure}

Interestingly for this model, the local and quantum bounds of the CHSH value are easily obtained by drawing a connection with a biased-inputs CHSH game~\cite{Lawson}, which we do in Appendix~\ref{bias}. It follows from this connection that, as long as $\varepsilon>1-1/\sqrt{2}\approx{}0.293$, the observation of the quantum bound self-tests a maximally entangled two-qubit state. Indeed, up to numerical precision, the limit cases are manifest in Fig.~\ref{fig:dirg}: $G_{0}(S_{\rm obs})=1$ only in the local regime, while $G_{0}(S_{\rm obs})=1/2$ if the quantum bound is saturated. For $\varepsilon\leq{}1-1/\sqrt{2}$, the quantum and local bounds coincide, such that no randomness can be certified based on the CHSH value with quantum behaviors.

As a final reflection, we remark that considering binary leakage is fully general against an adversary that coarse-grains the leaked information into a pair of guesses of $X$ and $Y$. In fact, it seems reasonable that the optimal exploitation of any leakage channel is of this kind.

On a different level, while exploring the connection between input leakage and input bias, we derive an analytical formula of Eve's local guessing probability in a CHSH game where Alice's input is biased towards Eve's target. This formula is given in Appendix~\ref{bias}.

\section{DIQKD WITH INPUT LEAKAGE}\label{DIQKD}
Minor modifications of the previous analysis allow to assess the performance of DIQKD with input leakage. In DIQKD, it is customary to add a third input to one of the parties for the purpose of key generation~\cite{Acín2007,PironioDIQKD}. In particular, here we take $x\in\{0,1\}$ and $y\in\{0,1,2\}$, such that any round with $(x,y)\in\{0,1\}^{2}$ is a \emph{test round}, and any round with $(x,y)=(0,2)$ is a \emph{key round}. As in DIRG, the process is characterized by the observed behavior $\mathcal{P}$ under the i.i.d. assumption, which in the context of QKD implies \emph{collective attacks}. Under these circumstances, an asymptotic lower bound on the secret key-rate follows given an upper bound on Eve's guessing probability of Alice's output in the key rounds. Denoting this bound by $G_{02}(S_{\rm obs})$, one can asymptotically extract~\cite{Masanes,Primaatmaja}
\begin{equation}
r_{\infty}=-\log_{2}G_{02}(S_{\rm obs})-h(Q)
\end{equation}
secret key bits per round, where $Q$ denotes the bit error probability of the key rounds. The optimization problem to obtain $G_{02}(S_{\rm obs})$ reads
\begin{eqnarray}\label{program_2}
&&\max_{\displaystyle{\bigl\{q_{\lambda|uv},\mathcal{P}_{uv\lambda}\bigr\}}}\quad \sum_{u,v}q_{uv|02}\sum_{a}\sum_{\mathcal{P}_{uv\lambda}\in\mathcal{G}_{a}}q_{\lambda|uv}\mathcal{P}_{uv\lambda}(a|0,2)\nonumber \\ &&\textrm{s.t.}\hspace{.1cm}\sum_{a,b}\sum_{x,y=0,1}(-1)^{a+b+xy}\sum_{u,v}q_{uv|xy}\mathcal{P}_{uv}(a,b|x,y)=S_{\rm obs},\nonumber \\
&&\hspace{.6cm}\mathcal{P}_{uv\lambda}\in{}\mathcal{Q}\hspace{.1cm}\textrm{for all } u,v,\lambda,
\end{eqnarray}
where $\mathcal{P}_{uv}(a,b|x,y)=\sum_{\lambda}q_{\lambda|uv}\mathcal{P}_{uv\lambda}(a,b|x,y)$. This problem is analogous to the one for $G_{0}(S_{\rm obs})$ in the previous section, except from the fact that the target refers to the key rounds, and the CHSH value is only computed with the test rounds. After applying the same \emph{behavior reduction} as in Sec.~\ref{DIRG}, the problem is simplified to
\begin{eqnarray}\label{behavior_reduced_2}
&&\max_{\displaystyle{\bigl\{\mathcal{P}_{uv}\bigr\}}}\quad \sum_{u,v}q_{uv|02}\mathcal{P}_{uv}(0|0,2)\hspace{.4cm}\nonumber \\
&&\textrm{s.t.}\hspace{.1cm}\sum_{a,b}\sum_{x,y=0,1}(-1)^{a+b+xy}\sum_{u,v}q_{uv|xy}\mathcal{P}_{uv}(a,b|x,y)=S_{\rm obs},\nonumber \\
&&\hspace{.6cm}\mathcal{P}_{uv}\in{}\mathcal{Q}\hspace{.2cm}\textrm{for all } u,v.
\end{eqnarray}
A crucial subtlety captured by Eq.~(\ref{behavior_reduced_2}) is that, in the presence of input leakage, Eve can partially distinguish key rounds from test rounds. As an intuitive example, she could use strongly entangled states if the leakage anticipates a test round (to boost the CHSH value), and weakly entangled states if the leakage anticipates a key round (to boost her guessing probability).

For illustration purposes, in Fig.~\ref{fig:diqkd} we provide key-rate plots using the intermediate level $\mathcal{Q}^{1+\mathrm{AB}}$ of the NPA hierarchy. We assume a binary symmetric channel for $U$ and a ternary symmetric channel for $V$, with respective crossover probabilities $P(U\neq{}X)=\varepsilon_{\rm a}\leq{}1/2$ and $P(V\neq{}Y)=\varepsilon_{\rm b}\leq{}2/3$~\cite{alphabets}. As standard in DIQKD, in the top panel we consider the statistics of a maximally entangled photon-pair subject to depolarizing noise, such that $Q=\left(1-S_{\rm obs}/2\sqrt{2}\right)/2$ for each $S_{\rm obs}\in[2,2\sqrt{2}]$~\cite{PironioDIQKD,choice}. Alternatively, in the bottom panel we assume that each photon in the photon-pair undergoes a lossy channel with transmittance $\eta=10^{-\alpha_{\rm att}L/10}$, where $L$ (km) denotes the distance between the source and the equidistant parties and $\alpha_{\rm att}$ (dB/km) denotes the attenuation coefficient of the channel. We set $\alpha_{\rm att}=0.2$ dB/km, which is a typical value for an optical fiber at 1550 nm. This model yields $S_{\rm obs}=2\sqrt{2}\eta^2+2(1-\eta)^2$ and $Q=\eta(1-\eta)$, under the standard assumption that the undetected events are locally mapped to a fixed detection outcome to close the detection loophole~\cite{Pearle,Gisin,PironioDIQKD}.
\begin{figure}[H]
    \centering
    \includegraphics[width=\columnwidth]{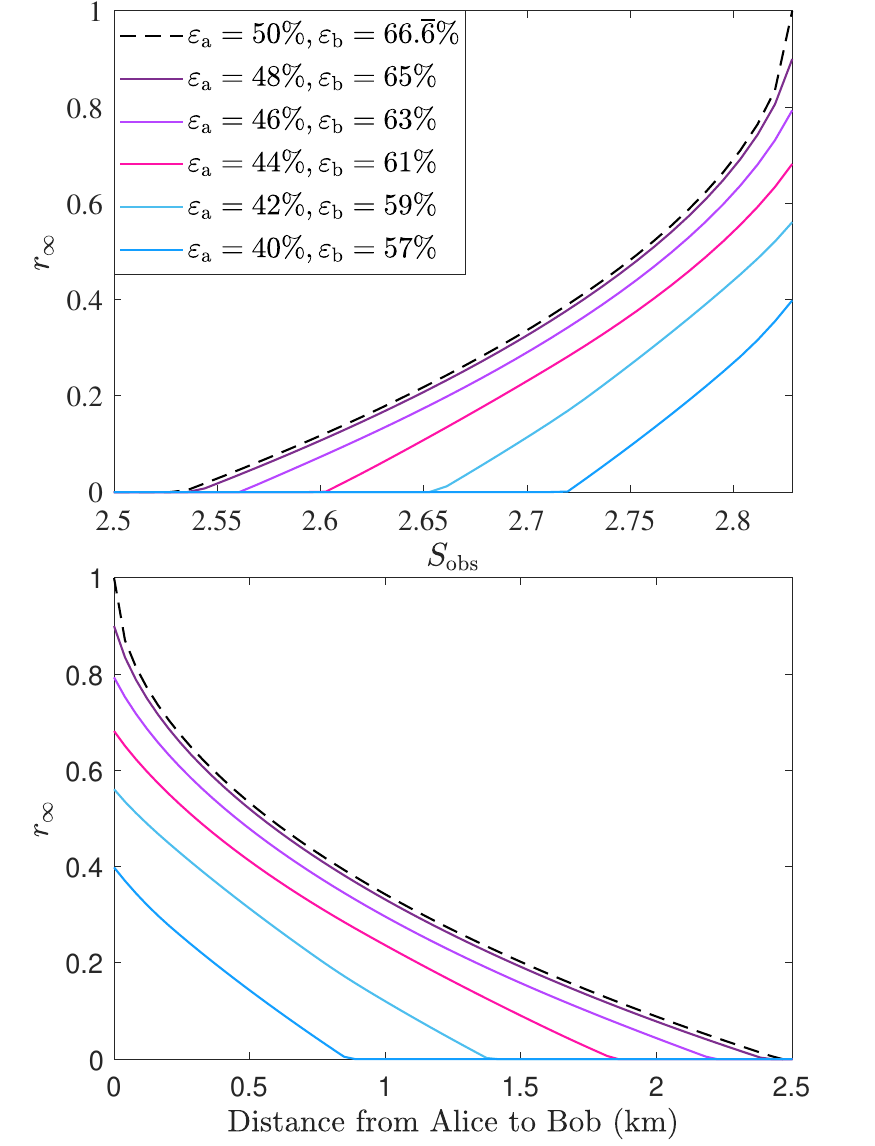}
    \caption{Lower bound on the asymptotic secret key-rate $r_{\infty}$ of a DIQKD protocol~\cite{PironioDIQKD} with input leakage. In each lab, the leakage is modeled with a $q$-ary symmetric channel, where $q$ denotes the number of local inputs (2 for Alice, 3 for Bob). The crossover parameters, $\varepsilon_{\rm a}$ and $\varepsilon_{\rm b}$, define the mismatch probabilities between the local inputs and their leaked values. Top panel: key rate as a function of the observed CHSH value, $S_{\rm obs}$, assuming a standard depolarizing noise model for the bit-error probability in the key-generation rounds. Bottom panel: key rate as a function of the distance between the parties, considering a central source that generates pairs of maximally entangled photons which traverse lossy (but noiseless) optical fibers of the same length. The reader is referred to the main text for further details.}
    \label{fig:diqkd}
\end{figure}

\section{Conclusions}\label{Outlook}
In this work, we have put forward simple analyses of device-independent randomness generation (DIRG) and device-independent quantum key distribution (DIQKD) with partial leakage of the protocol inputs, modeling the leaked information with characterized noisy channels. For a plug-and-play leakage model, we upper-bound Eve's local guessing probability or lower-bound the asymptotic key rate of a DIQKD protocol as a function of the magnitude of the input leakage.

One could argue that, even if this type of leakage surely takes place in Bell experiments, its active exploitation by the users' devices is the signature of a malicious design~\cite{Pironio2013,memory_attacks,malicious_1,malicious_2}. From a cryptographic perspective, however, this adversarial viewpoint seems inevitable to safely account for input leakage in DI cryptography.

Several extensions of this work are worth exploring. For practical purposes, it would be interesting to extend our analyses to the finite-size non-i.i.d.~regime via entropy accumulation~\cite{EAT1,EAT2,REAT}, possibly tackling the single-round optimization directly in terms of Von Neumann or Rényi entropies~\cite{Brown2024,Dupuis2023,Hahn2025}. Similarly, another natural direction would be to relax the characterization of the leakage channels, e.g. by only imposing boundary constraints on the transition probabilities $q_{u|x}$ and $q_{v|y}$. In addition, our results extend naturally to certify global (rather than local) randomness, which is the relevant resource for DI randomness amplification. Lastly, the proposed analyses can be extended to multipartite scenarios and different Bell inequalities.

\section{Acknowledgments} This work was supported by the Galician Regional Government (consolidation of research units: atlanTTic), the Spanish Ministry of Economy and Competitiveness (MINECO), the Fondo Europeo de Desarrollo Regional (FEDER) through the grant No. PID2024-162270OB-I00, the ``Hub Nacional de Excelencia en Comunicaciones Cuánticas” funded by the Spanish Ministry for Digital Transformation and the Public Service and the European Union NextGenerationEU, the European Union’s Horizon Europe Framework Programme under the Marie Sklodowska-Curie Grant No. 101072637 (Project QSI) and the project ``Quantum Secure Networks Partnership” (QSNP, grant agreement No 101114043), the European Union via the European Health and Digital Executive Agency (HADEA) under the Project QuTechSpace (grant 101135225), the European Union under the Project IberianQCI (grant 101249593), and the ``Programa de Cooperación Interreg VI-A España–Portugal" (POCTEP) 2021–2027 through the project QUANTUM IBER\_IA.


\appendix

\section{Hidden-variable models}\label{HVM}
Here, we list the three standard assumptions that define a local hidden-variable model. To stick to the literature on the topic, we avoid the behavior notation introduced in the main text.

Let $\Lambda$ be an arbitrary random variable possibly affecting the conditional statistics $p(a,b|x,y)$ of a CHSH test. We assume $\Lambda$ to be discrete for ease of notation, since contemplating continuous $\Lambda$'s does not affect any of our considerations. From the law of total probability, we have
\begin{equation}\label{sep}
p(a,b|x,y)=\sum_{\lambda}p(\lambda|x,y)p(a,b|x,y,\lambda),
\end{equation}
and various classes of hidden-variable models are defined by imposing constraints on $p(a,b|x,y,\lambda)$ and $p(\lambda|x,y)$. To be precise, $\Lambda$ is a \emph{measurement-independent} model if
\begin{equation}\label{mi}
p(\lambda|x,y)=p(\lambda),
\end{equation}
an \emph{outcome-independent} model if
\begin{equation}\label{OI}
p(a,b|x,y,\lambda)=p(a|x,y,\lambda)p(b|x,y,\lambda),
\end{equation}
and a \emph{parameter-independent} model if
\begin{equation}\label{PI}
p(a|x,y,\lambda)=p(a|x,\lambda)\hspace{.2cm}\mathrm{and}\hspace{.2cm}p(b|x,y,\lambda)=p(b|y,\lambda).
\end{equation}
Note that the parameter-independence of a hidden-variable model does not presume its outcome-independence in any way. To emphasize this, it is standard to write down Eq.~(\ref{PI}) as $\sum_{b}p(a,b|x,y,\lambda)=p(a|x,\lambda)$ and $\sum_{a}p(a,b|x,y,\lambda)=p(b|y,\lambda)$. Putting it all together, $\Lambda$ is a \emph{local} model of the statistics if it is measurement-independent, outcome-independent and parameter-independent, such that
\begin{equation}
p(a,b|x,y)=\sum_{\lambda}p(\lambda)p(a|x,\lambda)p(b|y,\lambda).
\end{equation}

\section{Semidefinite programs in the 1+AB level of the NPA hierarchy}\label{intermediate}
In this Appendix, we provide an explicit semidefinite program that computes $G_{0}(S_{\rm obs})$ in the presence of input leakage, resorting to the 1+AB level of the NPA hierarchy~\cite{NPA1,NPA2}.

Let $\mathcal{P}\in\mathcal{Q}$. From Proposition 4 in~\cite{NPA2}, $\mathcal{P}$ admits a \emph{pure} and \emph{projective} realization $\bigl(\ket{\psi},\{A_{a|x}\},\{B_{b|y}\}\bigr)$ in the commuting-operator framework~\cite{NPA1,NPA2}, with the property that the moment matrix $\Gamma$ associated to the list
\begin{align}
&\mathcal{F}_{1+\mathrm{AB}} = \{\mathds{1}, A_{0|0}, A_{0|1}, B_{0|0}, B_{0|1}, A_{0|0}B_{0|0}, A_{0|0}B_{0|1}, \nonumber \\
&A_{0|1}B_{0|0}, A_{0|1}B_{0|1}\},
\end{align}
is real, symmetric and positive semidefinite. The existence of such a $\Gamma$, together with the finite set of constraints imposed on it by (i) the idempotence of the local projectors and (ii) their commutation relations, provides a definition of the membership of $\mathcal{P}$ in $\mathcal{Q}^{1+\mathrm{AB}}$ in projector form~\cite{other_constraints}. 

Therefore, in the standard moment matrix notation $\Gamma(f,g)=\bra{\psi}f^{\dagger}g\ket{\psi}$~\cite{Tavakoli}, the relaxation of Eq.~(\ref{behavior_reduced}) at the 1+AB level reads
\begin{equation}\begin{split}\label{behavior_reduced_relaxed}
&\max\quad \sum_{u,v}q_{uv|x^{*}}\Gamma_{uv}\bigl(A^{uv}_{0|x^{*}},\mathds{1}\bigr)\hspace{.4cm}\\
&\hspace{.2cm}\textrm{s.t.}\hspace{.2cm}\sum_{a,b,x,y}(-1)^{a+b+xy}\sum_{u,v}q_{uv|xy}\Gamma_{uv}\bigl(A^{uv}_{a|x},B^{uv}_{b|y}\bigr)=S_{\rm obs},\\
&\hspace{.9cm}\Gamma_{uv}\succeq{}0\hspace{.2cm}\textrm{for all } u,v,
\end{split}\end{equation}
where we assume that the idempotence and commutation constraints are already incorporated in the $\Gamma_{uv}$ matrices. For completeness, this explicit encoding of the constraints in the $\Gamma_{uv}$ is worked out in the following proposition, which is somewhat inspired by Criterion 5 in~\cite{NPA2}.
\begin{proposition}[Membership in $\mathcal{Q}^{1+\mathrm{AB}}$]
$\mathcal{P}\in{}\mathcal{Q}^{1+\mathrm{AB}}$ if and only if there exists a real symmetric positive semidefinite matrix $\Gamma\succeq{}0$ of the form
\begin{equation}\label{certificate}
\Gamma= \\
\begin{pmatrix}
1      & a_0   & a_1   & b_0   & b_1   & c_{00} & c_{01} & c_{10} & c_{11} \\
a_0    & a_0   & \alpha& c_{00}& c_{01}& c_{00} & c_{01} & \mu_0  & \mu_1  \\
a_1    & \alpha& a_1   & c_{10}& c_{11}& \mu_0  & \mu_1  & c_{10} & c_{11} \\
b_0    & c_{00}& c_{10}& b_0   & \beta & c_{00} & \nu_0  & c_{10} & \nu_1  \\
b_1    & c_{01}& c_{11}& \beta & b_1   & \nu_0  & c_{01} & \nu_1  & c_{11} \\
c_{00} & c_{00}& \mu_0 & c_{00}& \nu_0 & c_{00} & \nu_0  & \mu_0  & \omega_0\\
c_{01} & c_{01}& \mu_1 & \nu_0 & c_{01}& \nu_0  & c_{01} & \omega_1& \mu_1 \\
c_{10} & \mu_0 & c_{10}& c_{10}& \nu_1 & \mu_0  & \omega_1& c_{10} & \nu_1 \\
c_{11} & \mu_1 & c_{11}& \nu_1 & c_{11}& \omega_0& \mu_1  & \nu_1  & c_{11}
\end{pmatrix},
\end{equation}
where $a_{x}=\mathcal{P}(0|x)$, $b_{y}=\mathcal{P}(0|y)$, $c_{xy}=\mathcal{P}(00|xy)$, and $\alpha$, $\beta$, $\mu_{0}$, $\mu_{1}$, $\nu_{0}$, $\nu_{1}$, $\omega_{0}$ and $\omega_{1}$ are unspecified entries.
\end{proposition}
\begin{proof}
It suffices to show that the ordered list of monomials in $\mathcal{F}_{1+\mathrm{AB}}$, together with the idempotence and commutation relations, exactly determine the structure of Eq.~(\ref{certificate}) on the moment matrix. For this, we introduce the notation $A^{+}_{x}=A_{0|x}$ and $B^{+}_{y}=B_{0|y}$, and define two lists of entries. First, we have 8 observable moments: $a_{x}=\Gamma(\mathds{1},A^{+}_{x})$, $b_{y}=\Gamma(\mathds{1},B^{+}_{y})$ and $c_{xy}=\Gamma(\mathds{1},A^{+}_{x}B^{+}_{y})$ for $x,y=0,1$. These provide a standard choice of 8 independent probabilities fully specifying $\mathcal{P}$. Secondly, we have 8 independent unobservable moments: $\alpha=\Gamma(A^{+}_{0},A^{+}_{1})$, $\beta=\Gamma(B^{+}_{0},B^{+}_{1})$, $\mu_{y}=\Gamma(A^{+}_{0},A^{+}_{1}B^{+}_{y})$, $\nu_{x}=\Gamma(B^{+}_{0},A^{+}_{x}B^{+}_{1})$, $\omega_{0}=\Gamma(A^{+}_{0}B^{+}_{0},A^{+}_{1}B^{+}_{1})$ and $\omega_{1}=\Gamma(A^{+}_{0}B^{+}_{1},A^{+}_{1}B^{+}_{0})$ for $x,y=0,1$. The independence of these 8 can be assessed as follows. First, the groups $\{\alpha,\beta\}$, $\{\mu_{0},\mu_{1},\nu_{0},\nu_{1}\}$ and $\{\omega_{0},\omega_{1}\}$ involve different minimum lengths (2, 3 and 4 operators, respectively), so it suffices to establish independence within each group~\cite{clarification2}. This is straightforward for the first two groups, and for the last group it follows from the fact that $\omega_{0}-\omega_{1}=\bra{\psi}A^{+}_{0}A^{+}_{1}[B^{+}_{0},B^{+}_{1}]\ket{\psi}$, and $[B^{+}_{0},B^{+}_{1}]$ is unconstrained. Given the independence of the 16 selected entries, the proof concludes by showing that they fix all other 29 diagonal or superdiagonal entries according to Eq.~(\ref{certificate}). The counting goes as follows. $\Gamma(\mathds{1},\mathds{1})=1$ (1 entry). $\Gamma(f,f)=\Gamma(\mathds{1},f)$ for all $f$, such that the diagonal of $\Gamma$ matches the list of observable moments (8 entries). $\Gamma(A^{+}_{x},B^{+}_{y})=\Gamma(A^{+}_{x},A^{+}_{x}B^{+}_{y})=\Gamma(B^{+}_{y},A^{+}_{x}B^{+}_{y})=c_{xy}$ for all $x,y$ (12 entries). $\Gamma(A^{+}_{1},A^{+}_{0}B^{+}_{y})=\Gamma(A^{+}_{0}B^{+}_{y},A^{+}_{1}B^{+}_{y})=\mu_{y}$ for all $y$ (4 entries). $\Gamma(B^{+}_{1},A^{+}_{x}B^{+}_{0})=\Gamma(A^{+}_{x}B^{+}_{0},A^{+}_{x}B^{+}_{1})=\nu_{x}$ for all $x$ (4 entries). This reproduces Eq.~(\ref{certificate}).
\end{proof}

\section{Connection with biased CHSH games}\label{bias}
Here, we make explicit the connection between input leakage and input bias in the CHSH scenario, and provide a complementary result for the latter case.

Let us consider that, in every round of the test, Alice and Bob draw their inputs from an arbitrary distribution $\pi_{xy}$, where we stick to the convention $a,b,x,y\in\{0,1\}$. For every choice of $\pi_{xy}$, a biased CHSH game can be defined by the single-round score~\cite{Brunner_review,Scarani_book}
\begin{equation}\label{rule1}
V_{abxy}=4(-1)^{a+b+xy},
\end{equation}
which encodes the fact that, in every round, the winning condition is $a+b=xy$ (mod 2), and a win/loss translates into a score of $\pm{}4$. Given an arbitrary behavior $\mathcal{P}$, its average score in this game reads
\begin{equation}\label{score}
V_{\pi}[\mathcal{P}]=4\sum_{a,b,x,y}\pi_{xy}\mathcal{P}(a,b|x,y)(-1)^{a+b+xy},
\end{equation}
which for $\pi_{xy}=1/4$ yields the definition of $S[\mathcal{P}]$ given in Eq.~(\ref{def}). That is to say,
\begin{equation}\label{equivalence}
S[\mathcal{P}]=V_{U}[\mathcal{P}],
\end{equation}
where $V_{U}$ is the linear functional that maps any behavior to its average score in an unbiased CHSH game.

In what follows, we upper bound $V_{U}[\mathcal{P}]$ (and thus $S[\mathcal{P}]$) in the presence of the type of input leakage $(U,V)$ described in the main text, characterized by the transition probabilities $q_{u|x}$ and $q_{v|y}$. Expanded in terms of the underlying behaviors $\mathcal{P}_{uv}$, $V_{U}[\mathcal{P}]$ reads
\begin{equation}\label{score2}
V_{U}[\mathcal{P}]=4\sum_{a,b,x,y,u,v}\pi_{xyuv}\mathcal{P}_{uv}(a,b|x,y)(-1)^{a+b+xy}
\end{equation}
for $\pi_{xyuv}=q_{uv}\pi_{xy|uv}$, where $q_{uv}$ and $\pi_{xy|uv}$ are fixed via Bayes' theorem given the transition probabilities and the uniform-inputs rule. Grouping conveniently, we have that
\begin{equation}\label{score3}
V_{U}[\mathcal{P}]=\sum_{u,v}q_{uv}\left[\hspace{.1cm}4\sum_{a,b,x,y}\pi_{xy|uv}\mathcal{P}_{uv}(a,b|x,y)(-1)^{a+b+xy}\right],
\end{equation}
where the term between brackets is the average score of $\mathcal{P}_{uv}$ in a biased CHSH game with input distribution $\pi_{xy|uv}=\pi_{x|u}\pi_{y|v}$. Unsurprisingly, every leaked pair $(u,v)$ determines a CHSH game with different input biases, quantified by the independent distributions $\pi_{x|u}$ and $\pi_{y|v}$. For any such game, the maximum score attainable with $\mathcal{P}_{uv}\in{}\mathcal{Q}$ is known to be~\cite{Lawson}
\begin{equation}\label{score_Q}
\begin{split}
&V^{uv}_{\mathcal{Q}}= \\
&\left\{
\begin{array}{ll}
4\sqrt{2}\sqrt{r_{u}^2+(1-r_{u})^2}\sqrt{s_{v}^2+(1-s_{v})^2} & \mathrm{if}\hspace{.2cm}r_{u}s_{v}\leq{}1/2 \\
4-8(1-r_{u})(1-s_{v}) & \rm{otherwise,} \\
\end{array} 
\right.
\end{split}
\end{equation}
where $r_{u}=\max_{x}\{\pi_{x|u}\}$ and $s_{v}=\max_{y}\{\pi_{y|v}\}$. In fact, for $r_{u}s_{v}\geq{}1/2$, $V^{uv}_{\mathcal{Q}}$ coincides with the maximum local score. Plugging Eq.~(\ref{score_Q}) into Eq.~(\ref{score3}), we conclude that $V_{U}[\mathcal{P}]\leq{}\sum_{uv}q_{uv}V^{uv}_{\mathcal{Q}}$. This result holds for any choice of $q_{u|x}$ and $q_{v|y}$, with arbitrary alphabets for $U$ and $V$.

Particularizing to the case of binary symmetric leakage with crossover probabilities $\varepsilon_{\rm a}$ and $\varepsilon_{\rm b}$, all pairs $(u,v)$ happen to reach the same $V^{uv}_{\mathcal{Q}}$, such that
\begin{equation}
S[\mathcal{P}]\leq{}4\sqrt{2}\sqrt{\vphantom{(1-\varepsilon_{\rm b})^2}(1-\varepsilon_{\rm a})^2+\varepsilon_{\rm a}^2}\sqrt{(1-\varepsilon_{\rm b})^2+\varepsilon_{\rm b}^2}
\end{equation}
if $(1-\varepsilon_{\rm a})(1-\varepsilon_{\rm b})\leq{}1/2$, and $S[\mathcal{P}]\leq{}4-8\varepsilon_{\rm a}\varepsilon_{\rm b}$ otherwise.

As a complementary result, in the exploration of the link between input leakage and input bias, we derive an analytical formula for Eve's guessing probability in a biased CHSH game, in the restricted case where (i) Alice's input is biased towards Eve's target $x^{*}$, and (ii) Bob's input is uniformly random. That is to say, $\pi_{xy}=\pi_{x}\pi_{y}$ for $\pi_{x}=r\delta_{x,0}+(1-r)\delta_{x,1}$ and $\pi_{y}=1/2$ for all $y$. In the nontrivial quantum regime where $S_{\rm obs}\in\bigl[4r,4\sqrt{r^{2}+(1-r)^2}\bigr]$, the precise bound reads
\begin{equation}\label{one-sided}
G_{x^{*}}(S_{\rm obs})=\frac{1}{2}\left\{  1 + \sqrt{ 1-\frac{1}{(1-r)^2}\left[\left(\frac{S_{\rm obs}}{4}\right)^2-r^2\right] }  \right\},
\end{equation}
which reduces to Eq.~(\ref{pguess}) for $r=1/2$. A simple way to obtain Eq.~(\ref{one-sided}) is to notice that, for the input distribution $\pi_{xy}$ under consideration,
\begin{equation}\label{connection}
V_{\pi}[\mathcal{P}]=\frac{1-r}{2}I^{0}_{\frac{r}{1-r}}[\mathcal{P}],
\end{equation}
where $I^{\beta}_{\alpha}$ is the Bell inequality considered in~\cite{Acín2012}. With this in mind, the desired analytical formula can be obtained from the state-dependent quantum bound of Eq.~(\ref{connection}) for an arbitrary pure entangled state of two qubits. This derivation is rather standard and formally identical to that of Eq.~(\ref{pguess}), so we omit it here for brevity.


\begin{thebibliography}{50}
\bibitem{lightcones}
In the literature on hidden-variable models, one often requests that the hidden variables lie in the backward light-cones of the test inputs. This rules out the possibility of ``trivial" measurement-dependent hidden variables originating in the causal future of the inputs, which would defeat the purpose that most theorists have in mind with hidden-variable models. We do not make this request here, because input leakage indeed belongs to this ``trivial" kind.
\bibitem{clarification}
If $X$ and $Y$ were correlated to $\Lambda$, $q_{\lambda|xy}$ would generally differ from $q_{\lambda}$, thus invalidating the behavior equality of Eq.~(\ref{convex}) (needless to say, for arbitrary $\mathcal{R}$ and $\mathcal{S}$, $\mathcal{R}=\mathcal{S}$ asserts the entry-wise identification of the two behaviors).
\bibitem{caveat}
Eq.~(\ref{local}) is customarily used as a characterization of locality, rather than a definition. The actual definition is the following: a behavior $\mathcal{P}$ is called local if it admits a local hidden-variable model (as defined in Appendix~\ref{HVM}). 
\bibitem{crucial}
Observe that $\mathcal{P}_{uv}=\sum_{\lambda}q_{\lambda|uv}\mathcal{P}_{uv\lambda}$ but $\mathcal{P}\neq{}\sum_{uv}q_{uv}\mathcal{P}_{uv}$, because $(U,V)$ is not independent of $(X,Y)$.
\bibitem{alphabets}
Note that the mismatch probability between a local input and its leaked value is only defined if the output alphabet of the leakage channel contains the input alphabet.
\bibitem{choice}
In the depolarizing noise model, the relation $Q=\left(1-S_{\rm obs}/2\sqrt{2}\right)/2$ holds if the measurements are those saturating the Tsirelson's bound with a maximally entangled state.
\bibitem{other_constraints}
For the CHSH scenario, the quantum-mechanical constraints on the independent local projectors include hermiticity, idempotence and commutation relations~\cite{NPA2}. Nevertheless, the symmetry of the moment matrix already encodes the hermiticity constraints on it.
\bibitem{clarification2}
The automatic independence across groups uses the fact that no completeness relations apply on the operators of $\mathcal{F}_{1+\mathrm{AB}}$, which guarantees that an arbitrary monomial with minimum length $k$ cannot be expressed as a sum of monomials with minimum length $k+1$ (or larger).
\end{thebibliography}
\end{document}